\documentclass[10pt, letterpaper, conference]{ieeeconf}
\IEEEoverridecommandlockouts
\usepackage{amsmath, amsfonts, amssymb, mathtools, url, graphicx, graphicx}
\usepackage{newtxtext}
\usepackage{newtxmath}
\let\labelindent\relax
\usepackage{enumitem}
\usepackage[usenames, dvipsnames]{color}

\allowdisplaybreaks
\usepackage{tikz}
\usetikzlibrary{automata,positioning}
\usepackage{caption}
\usepackage{subcaption}
\usepackage{algorithm,algorithmic}

\newtheorem{lemma}{Lemma}
\newtheorem{proposition}{Proposition}
\newtheorem{defn}{Definition}
\newtheorem{assump}{Assumption}

\newtheorem{prob}{Problem}
\newtheorem{remark}{Remark}

\newtheorem{fact}{Fact}

\newcommand{\norm}[1]{\left\lVert{#1}\right\rVert}
\newcommand{\abs}[1]{\left\lvert{#1}\right\rvert}

\newcommand{\pmat}[1]{\begin{pmatrix}#1\end{pmatrix}}

\renewcommand{\geq}{\geqslant}

\renewcommand{\leq}{\leqslant}

\newcommand{\R}{\mathbb{R}}
\newcommand{\N}{\mathbb{N}}
\renewcommand{\P}{\mathcal{P}}
\newcommand{\M}{\mathcal{M}}

\DeclareMathOperator{\minimize}{minimize}
\DeclareMathOperator{\sbjto}{subject\;to}

\allowdisplaybreaks

\title{Data-based computation of stabilizing minimum dwell times\\for discrete-time switched linear systems}
\author{Atreyee Kundu%
	\thanks{The author is with the Department of Electrical Engineering, Indian Institute of Science Bangalore, Bengaluru - 560012, Karnataka, India, E-mail: \texttt{atreyeek@iisc.ac.in}}
	\thanks{Atreyee Kundu's research work is supported by INSPIRE Faculty Award IFA17-ENG225 by the Department of Science and Technology, Govt. of India.}
}

\begin{document}

	\maketitle

	\begin{abstract}
        We present an algorithm to compute stabilizing minimum dwell times for discrete-time switched linear systems without the explicit knowledge of state-space models of their subsystems. Given a set of finite traces of state trajectories of the subsystems that satisfies certain properties, our algorithm involves the following tasks: first, multiple Lyapunov functions are designed from the given data; second, a set of relevant scalars is computed from these functions; and third, a stabilizing minimum dwell time is determined as a function of these scalars. A numerical example is presented to demonstrate the proposed algorithm.
	\end{abstract}

\section{Introduction}
\label{s:intro}
\subsection{Motivation}
\label{ss:motive}
    Identification of classes of switching signals that preserve stability of switched systems (commonly called as stabilizing switching signals) constitutes a key topic in hybrid systems literature. A vast body of the existing works in this direction utilizes the concept of ``slow switching'' vis-a-vis \emph{minimum dwell time switching} and its variants, see e.g., \cite[Chapter 3]{Liberzon}, \cite{Zhai'02,abc,ghi} for results and discussions. A switching signal satisfying a minimum dwell time \(\tau > 0\) dwells on any subsystem at least for \(\tau\) units of time before it switches to a different subsystem. \(\tau\) is called \emph{stabilizing} if a switching signal obeying this minimum dwell time preserves stability of the switched system under consideration. Loosely speaking, the underlying idea of stability under minimum dwell time switching is that if all the subsystems are stable and the switching is sufficiently slow, then the ``energy injected due to switching'' gets sufficient time for dissipation due to the stability of the individual subsystems. It follows that a switched system whose individual subsystems are stable, may be unstable under arbitrary switching signals but always admits a stabilizing switching signal that satisfies a large enough minimum dwell time. Numerical computation of stabilizing minimum dwell times for switched systems typically requires the availability of mathematical models of the subsystems, see e.g., \cite{abc,ghi} and the references therein. However, in many real-world scenarios, particularly for large-scale complex systems, accurate mathematical models such as transfer functions, state-space models or  kernel representations of the subsystems are often not present. This interesting fact motivates the current paper. We devise an algorithm to compute stabilizing minimum dwell times for discrete-time switched linear systems when explicit knowledge of the state-space models of their subsystems are not available.
\subsection{Literature survey}
\label{ss:lit_survey}
    Stability analysis and control synthesis of switched systems without explicitly involving mathematical models of their subsystems, are dealt with recently in \cite{Kenarian2019,Li2018,def}. The work \cite{Kenarian2019} addresses the problem of deciding stability of a discrete-time switched linear system from a set of finite traces of state trajectories. Probabilistic stability guarantees are provided as a function of the number of available state observations and a desired level of confidence. In \cite{Li2018} reinforcement learning techniques are employed for optimal control of switched linear systems. A Q-learning based algorithm is proposed to design a discrete switching signal and a continuous control signal such that a certain infinite-horizon cost function is minimized. The convergence guarantee of the proposed algorithm is, however, not available. A randomized polynomial-time algorithm for the design of switching signals under the availability of certain information about the multiple Lyapunov(-like) functions \cite[\S3.1]{Liberzon} corresponding to the individual subsystems in an expected sense, is presented in \cite{def}. The authors show that if it is allowed to switch from any subsystem to a certain number of stable subsystems, then a switching signal obtained from the proposed algorithm is stabilizing with overwhelming probability.
\subsection{Our contributions}
\label{ss:contri}
    We consider the availability of finite traces of state trajectories of the subsystems that satisfy certain properties (henceforth to be called as subsystems data), possibly collected from a simulation model or during the operation of the switched system, and combine two ingredients: (a) data-based techniques for stability analysis of discrete-time linear systems and (b) multiple Lyapunov functions based techniques for the computation of stabilizing minimum dwell times for switched systems, towards developing an algorithm for the computation of stabilizing minimum dwell times in the absence of explicit knowledge of the state-space models of the subsystems. Our computation of stabilizing minimum dwell times involves the following steps: first, we design multiple Lyapunov functions, one for each subsystem, from the given data; second, we compute a set of scalars from these functions; and third, we determine a stabilizing minimum dwell time as a function of the above set of scalars.

    Computation of stabilizing minimum dwell times for switched systems from multiple Lyapunov functions is standard, see e.g., \cite{Zhai'02,ghi} and the references therein. However, these functions are commonly designed under complete knowledge of the state-space models of the subsystems. We assume certain properties of the subsystems data, and involve techniques of data-based stability analysis of linear systems proposed in \cite{Park2009} to design multiple Lyapunov functions. These functions are then employed to compute minimum dwell times. At this point, it is worth highlighting that we do not opt for the construction of mathematical models of the subsystems from the given data, and hence the proposed technique does not involve system identification of the subsystems.\footnote{In general, the problem of system identification of switched systems involves identification of the subsystems dynamics from noisy input-output data collected during the operation of the system, see e.g., the recent works \cite{Petreczky2018,Gosea2018,Dabbene2019} and the references therein.} To the best of our knowledge, this is the first instance in the literature where stabilizing dwell times for switched systems are computed without the explicit knowledge of state-space models of the subsystems.
\subsection{Paper organization}
\label{ss:paper_org}
    The remainder of this paper is organized as follows: In \S\ref{s:prob_stat} we formulate the problem under consideration. A set of preliminaries required for our result is presented in \S\ref{s:prelims}. Our result appears in \S\ref{s:mainres}. We present a numerical example in \S\ref{s:numex}, and conclude in \S\ref{s:concln} with a brief discussion of future research directions.
\subsection{Notation}
\label{ss:notation}
    \(\R\) is the set of real numbers and \(\N\) is the set of natural numbers, \(\N_{0} = \N\cup\{0\}\). \(I_{d}\) denotes the \(d\times d\) identity matrix, \(0_{n}\) and \(\overline{0}_{n}\) denote the \(d\times 1\) zero matrix and \(d\times d\) zero matrix, respectively. For a matrix \(B\in\R^{d\times d}\), \(B\succ 0\) (resp., \(B\prec 0\)) denotes that \(B\) is positive definite (resp., negative definite), and \(\lambda_{\max}(B)\) denotes the maximal eigenvalue of \(B\).
\section{Problem statement}
\label{s:prob_stat}
    We consider a family of discrete-time linear systems
    \begin{align}
    \label{e:family}
        x(t+1) = A_{i}x(t),\:\:x(0) = x_{0},\:\:i\in\P,\:\:t\in\N_{0},
    \end{align}
    where \(x(t)\in\R^{d}\) is the vector of states at time \(t\), \(\P = \{1,2,\ldots,N\}\) is an index set, and
    \begin{align}
    \label{e:comp_form}
        A_{i} = \pmat{-a_{i,d-1}&\cdots&-a_{i,1}&-a_{i,0}\\1&\cdots&&0\\&\vdots&&\vdots\\&\cdots&1&0}\in\R^{d\times d},\:\:i\in\P
    \end{align}
    are full-rank constant Schur stable matrices.\footnote{A matrix \(M\in\R^{d\times d}\) is Schur stable if all its eigenvalues are inside the open unit disk. We call \(M\) unstable if it is not Schur stable.} Let \(\sigma:\N_{0}\to\P\) be a switching signal. A discrete-time switched linear system generated by the family of systems \eqref{e:family} and a switching signal \(\sigma\) is described as
    \begin{align}
    \label{e:swsys}
        x(t+1) = A_{\sigma(t)}x(t),\:\:x(0) = x_{0},\:\:t\in\N_{0},
    \end{align}
    where we have suppressed the dependence of \(x\) on \(\sigma\) for notational simplicity.

    The scalars \(a_{i,j}\), \(j=0,1,\ldots,d-1\), \(i\in\P\), are unknown. We let \(\chi = \{(x_{i}(0),x_{i}(1),\ldots,x_{i}(L)),\:i\in\P\}\) be a given set of finite traces of state trajectories of the subsystems \(i\in\P\). Here, \(x_{i}(T+1) = A_{i}x_{i}(T)\), \(T=0,1,\ldots,L-1\), \(L\in\N\). In the sequel we will refer to the set \(\chi\) as \emph{subsystems data}.

    Let \(0=:\kappa_{0}<\kappa_{1}<\kappa_{2}<\cdots\) be the \emph{switching instants}; these are the points in time where \(\sigma\) jumps. A switching signal \(\sigma\) is said to satisfy a minimum dwell time \(\tau > 0\) if the following condition holds:
    \begin{align}
    \label{e:dw_time}
        \kappa_{i+1} - \kappa_{i} \geq \tau,\:\:i=0,1,2,\ldots.
    \end{align}
    For a fixed \(\tau\), we let \(\mathcal{S}_\tau\) denote the set of all switching signals \(\sigma\) that satisfy condition \eqref{e:dw_time}. We are interested in stability of the switched system \eqref{e:swsys} under every \(\sigma\in\mathcal{S}_{\tau}\) for a certain \(\tau\). Recall that
    \begin{defn}
    The switched system \eqref{e:swsys} is \emph{globally asymptotically stable} (GAS) for a given switching signal $\sigma$ if \eqref{e:swsys} is Lyapunov stable and globally asymptotically convergent, i.e., for all $x(0)$, $\displaystyle{\lim_{t\rightarrow+\infty}\norm{x(t)} = 0}$.
    \end{defn}

     We will solve the following problem:
    \begin{prob}
    \label{prob:mainprob}
        Given \(\chi\), compute \(\tau\) such that the switched system \eqref{e:swsys} is GAS under every switching signal \(\sigma\in\mathcal{S}_{\tau}\).
    \end{prob}

   Towards solving Problem \ref{prob:mainprob}, we will assume certain properties of \(\chi\), design multiple Lyapunov functions for the subsystems \(i\in\P\) from \(\chi\), and compute a set of scalars corresponding to these functions. Then we will compute a dwell time \(\tau\) as a function of the above set of scalars such that each element \(\sigma\in\mathcal{S}_{\tau}\) is stabilizing. Prior to presenting our solution to Problem \ref{prob:mainprob}, we catalog a set of preliminaries.
\section{Preliminaries}
\label{s:prelims}
    The following fact is well-known:
    \begin{fact}{\cite[Fact 1]{abc}}
    \label{fact:key1}
        For each \(i\in\P\), there exists a pair \((P_{i},\lambda_{i})\in\R^{d\times d}\times\R\), where \(P_{i}\) is a symmetric and positive definite matrix and \(0 < \lambda_{i} < 1\), such that, with
        \begin{align}
        \label{e:Lyap_eq1}
            \R^{d}\ni\xi\mapsto V_{i}(\xi) := \xi^\top P_{i}\xi\in[0,+\infty[,
        \end{align}
        we have
        \begin{align}
        \label{e:Lyap_eq2}
            V_{i}(\gamma_{i}(t+1))\leq\lambda_{i}V_{i}(\gamma_{i}(t)),\:\:t\in\N_{0},
        \end{align}
        and \(\gamma_{i}(\cdot)\) solves the \(i\)-th recursion in \eqref{e:family}.
    \end{fact}
    The functions \(V_{i}\), \(i\in\P\) are Lyapunov functions corresponding to the subsystems \(i\in\P\). The scalar \(\lambda_{i}\), \(i\in\P\) gives a quantitative measure of stability of subsystem \(i\). The Lyapunov functions corresponding to the individual subsystems are related as follows:
    \begin{fact}{\cite[Fact 2]{abc}}
    \label{fact:key2}
        There exists \(\R\ni\mu_{ij}>0\) such that
        \begin{align}
        \label{e:Lyap_eq3}
            V_{j}(\xi)\leq\mu_{ij}V_{i}(\xi)\:\:\text{for all}\:\xi\in\R^{d},\:i,j\in\P,\:i\neq j.
        \end{align}
    \end{fact}
    A \emph{tight} estimate of \(\mu_{ij}\), \(i,j\in\P\), is provided below.
    \begin{proposition}{\cite[Proposition 1]{abc}}
    \label{prop:mu_estimate}
        The scalars \(\mu_{ij}\), \(i,j\in\P\) can be computed as follows:
        \begin{align}
        \label{e:mu_estimate}
            \mu_{ij} = \lambda_{\max}(P_{j}P_{i}^{-1}),\:\:i,j\in\P.
        \end{align}
    \end{proposition}

    \begin{lemma}
    \label{lem:auxres1}
        Consider the family of systems \eqref{e:family}. Let \(\lambda_{i}=\lambda_{s}\) for all \(i\in\P\) and \(\mu_{ij} = \mu\) for all \(i,j\in\P\). Then the switched system \eqref{e:swsys} is GAS for every switching signal \(\sigma\in S_{\tau}\) with
        \begin{align}
        \label{e:tau_estimate}
            \tau > \frac{\ln\mu}{\abs{\ln\lambda_{s}}}.
        \end{align}
    \end{lemma}

    The estimate of a stabilizing minimum dwell time, \(\tau\), presented in Lemma \ref{lem:auxres1} is standard in the literature, and has been proved in many contexts. For example, Lemma \ref{lem:auxres1} follows directly from \cite[Proposition 1]{ghi} with no unstable subsystems.

    The computation of the scalars \(\lambda_{i}\), \(i\in\P\) and \(\mu_{ij}\), \(i,j\in\P\) (and hence the scalars \(\lambda_{s}\) and \(\mu\)) with known matrices, \(A_{i}\), \(i\in\P\), is addressed in \cite{abc}. However, we do not have the said information available. To cater to this scenario, we rely on data-based computation techniques of quadratic Lyapunov functions for stable linear systems presented in \cite{Park2009}.

    Let \(x^{(p)}_{i}(T)\) denote the \(p\)-th element of the vector \(x_{i}(T)\), \(p=1,2,\ldots,d\), \(T=0,1,\ldots,L\), \(i\in\P\). We define
    \[
        q_{i}(T) = \pmat{x^{(1)}_{i}(T+1)\\x^{(1)}_{i}(T)\\\vdots\\x^{(d)}_{i}(T)},\:\:T=0,1,\ldots,L-1.
    \]
    Let
    \begin{align}
    \label{e:psi_defn}
        \Psi_{i} = \pmat{q_{i}(T) & q_{i}(T+1) & \cdots & q_{i}(T+d-1)}
    \end{align}
    \(\in\R^{(d+1)\times d},\:\:T\in\{0,1,\ldots,L-1\}\) be such that its column vectors are linearly independent. We will operate under the following
    \begin{assump}
    \label{assump:psi_exists}
    \rm{
        The set \(\chi\) is such that \(\Psi_{i}\) is well-defined for all \(i\in\P\).
    }
    \end{assump}
    Fix \(i\in\P\). Notice that whether \(\Psi_{i}\) is well-defined or not, depends on the initial value, \(x_{i}(0)\), and the length, \(L+1\), of the available trace \((x_{i}(0),x_{i}(1),\ldots,x_{i}(L))\). Indeed, to design \(\Psi_{i}\), we need \(L\geq d\) and \(x_{i}(0)\) is such that the vectors \(q_{i}(T),q_{i}(T+1),\ldots,q_{i}(T+d-1)\) are defined and linearly independent for some \(T\in\{0,1,\ldots,L-1\}\). From \cite[Lemma 3]{Park2004} it follows that for every \(A_{i}\) in the companion form described in \eqref{e:comp_form}, there exists \(x_{i}(0)\in\R^{d}\) such that \(\Psi_{i}\) is well-defined with \(T=0\) and \(L=d\). Clearly, a large number of traces of state trajectories of the subsystem \(i\) can be collected to arrive at such \(x_{i}(0)\).

    \begin{lemma}
    \label{lem:auxres2}
    \rm{
        For each subsystem \(i\in\P\), condition \eqref{e:Lyap_eq2} is equivalent to the following: there exists a symmetric and positive definite matrix \(P_{i}\in\R^{d\times d}\) and a scalar \(0<\lambda_{i}<1\) such that
        \begin{align}
        \label{e:key_ineq}
            \Psi_{i}^\top\pmat{I_{n} & 0_{n}\\0_{n} & I_{n}}^\top\pmat{P_{i} & \overline{0}_{n}\\\overline{0}_{n} & -\lambda_{i}P_{i}}\pmat{I_{n} & 0_{n}\\0_{n} & I_{n}}\Psi_{i} \prec 0.
        \end{align}
    }
    \end{lemma}

    \begin{proof}
        Follows under the set of arguments employed in \cite[Theorem 2]{Park2009}.
    \end{proof}

    Lemma \ref{lem:auxres2} provides a mechanism to design a set of Lyapunov functions \(V_{i}\), \(i\in\P\), defined in Fact \ref{fact:key1}, from the subsystems data, \(\chi\). We now combine Lemmas \ref{lem:auxres1}, \ref{lem:auxres2} and Proposition \ref{prop:mu_estimate} to provide a solution to Problem \ref{prob:mainprob}.
\section{Result}
\label{s:mainres}

    Given subsystems data, \(\chi\), such that Assumption \ref{assump:psi_exists} holds, Algorithm \ref{algo:dwell_time} computes stabilizing minimum dwell times for discrete-time switched linear systems. It involves the following tasks:
    \begin{itemize}[label = \(\circ\), leftmargin = *]
        \item First, the matrices \(\Psi_{i}\), \(i\in\P\), are constructed from the set \(\chi\).
        \item Second, Lemma \ref{lem:auxres2} is employed to compute symmetric and positive definite matrices, \(P_{i}\), \(i\in\P\) and a scalar \(0<\lambda_{s}<1\) that satisfy \eqref{e:key_ineq} with \(\lambda_{i}=\lambda_{s}\) for all \(i\in\P\).
        \item Third, the estimates of \(P_{i}\), \(i\in\P\) obtained above are used to compute the scalars \(\mu_{ij}\), \(i,j\in\P\) by employing Proposition \ref{prop:mu_estimate}. The scalar \(\mu\) is chosen to be the maximum of \(\mu_{ij}\), \(i,j\in\P\).
        \item Fourth, \(\tau\) is computed as a function of \(\lambda_{s}\) and \(\mu\) as described in Lemma \ref{lem:auxres1}.
    \end{itemize}

    \begin{algorithm}[htbp]
	\caption{Model-free computation of stabilizing minimum dwell time, \(\tau\)} \label{algo:dwell_time}
    	\begin{algorithmic}[1]
    			\renewcommand{\algorithmicrequire}{\textbf{Input}:}
			\renewcommand{\algorithmicensure}{\textbf{Output}:}
	
			\REQUIRE Subsystems data, \(\chi\), such that Assumption \ref{assump:psi_exists} holds.
			\ENSURE A stabilizing minimum dwell time, \(\tau\).

            \STATE {\bf Step I}: Construct \(\Psi_{i}\), \(i\in\P\) from \(\chi\).
			
            \STATE {\bf Step II}: Compute \(P_{i}\), \(i\in\P\) and \(\lambda_{s}\) as follows:
                \STATE Fix \(h > 0\) (small enough) and compute \(k\in\N\) such that \(k\) is the largest integer satisfying \(kh<1\).
                \FOR {\(\lambda_{s} = h,2h,\ldots,kh\)}
                    \FOR {\(i=1,2,\ldots,N\)}
                    \STATE Solve the following feasibility problem in \(P_{i}\):
                    \begin{align}
	           		 		\label{e:feasprob}
		              				\minimize\:\:&\:\:1\nonumber\\
		              				\sbjto\:\:&\:\:
		                  			\begin{cases}
                              			\text{condition}\:\eqref{e:key_ineq}\:\text{with}\:\lambda_{i}=\lambda_{s},\\
                                        P_{i}^\top = P_{i} \succ 0.
		                  			\end{cases}
	             				\end{align}
                    \ENDFOR
                    \IF {a solution \(P_{i}\) to \eqref{e:feasprob} is found for all \(i\in\P\)}
                        \STATE Store \(\lambda_{s}\) and \(P_{i}\), \(i\in\P\), and go to Step II.
                    \ENDIF
                \ENDFOR
                				
            \STATE {\bf Step III}: Compute \(\mu\) from \(P_{i}\), \(i\in\P\) as follows:
                \FOR {\(i=1,2,\ldots,N\)}
                    \FOR {\(j=1,2,\ldots,N\)}
                        \STATE Set \(\mu_{ij} = \lambda_{\max}(P_{j}P_{i}^{-1})\).
                    \ENDFOR
                \ENDFOR
                \STATE Set \(\displaystyle{\mu = \max_{i,j\in\P}\mu_{ij}}\).
            \STATE {\bf Step IV}: Compute \(\tau\) as follows:
                \STATE Pick \(\varepsilon > 0\) (small enough).
                \STATE Set \(\tau = \lceil\frac{\ln\mu}{\abs{\ln\lambda_{s}}}+\varepsilon\rceil\).
					 	
    		\end{algorithmic}
   \end{algorithm}

    Notice that solving \eqref{e:key_ineq} with both \(P_{i}\), \(i\in\P\) and \(\lambda_{s}\) unknown is a numerically difficult task. To address this issue, we employ a line search technique \cite{LMI_book} as follows: a finite set of values of \(\lambda_{s}\) on the interval \(]0,1[\) is fixed, and corresponding to each element of this set, the feasibility problem \eqref{e:feasprob} is solved for \(N\) symmetric and positive definite matrices, \(P_{i}\), \(i\in\P\). The value of \(\lambda_{s}\) for which \eqref{e:feasprob} admits a solution for all \(i\in\P\) and its corresponding \(P_{i}\), \(i\in\P\) are stored. We observe the following:
    \begin{proposition}
    \label{prop:mu_1}
    \rm{
        Consider \(\mu\in\R\) computed as \(\displaystyle{\mu = \max_{i,j\in\P}\mu_{ij}}\), where \(\mu_{ij} \), \(i,j\in\P\) are as given in \eqref{e:mu_estimate}. Then \(\mu > 1\).
    }
    \end{proposition}
    \begin{proof}
        Fix \(i,j\in\P\), \(i\neq j\). Let \(0 < \mu_{ij} = \lambda_{\max}(P_{j}P_{i}^{-1}) < 1\). Clearly, \(\lambda_{\min}(P_{j}P_{i}^{-1})\leq\lambda_{\max}(P_{j}P_{i}^{-1})<1\). In addition, \(P_{j}P_{i}^{-1}\) is similar to \(P_{i}^{-1/2}(P_{j}P_{i}^{-1})P_{i}^{1/2}\) and the matrix \(P_{i}^{-1/2}P_{j}P_{i}^{-1/2}\) is symmetric and positive definite. Since the spectrum of a matrix is invariant under similarity transformations, we have \(\lambda_{\min}(P_{j}P_{i}^{-1}) > 0\).

        Now, \(\mu_{ji} = \lambda_{\max}(P_{i}P_{j}^{-1}) = \lambda_{\max}((P_{j}P_{i}^{-1})^{-1})=\frac{1}{\lambda_{\min}(P_{j}P_{i}^{-1})}\). Since \(0 < \lambda_{\min}(P_{j}P_{i}^{-1}) < 1\), it follows that \(\mu_{ji} > 1\). Consequently, \(\mu > 1\).
    \end{proof}
    Proposition \ref{prop:mu_1} asserts that \(\frac{\ln\mu}{\abs{\ln\lambda_{s}}} > 0\).

   \begin{proposition}
   \label{prop:mainres}
        Consider the switched system \eqref{e:swsys}. Suppose that subsystems data, \(\chi\), such that Assumption \ref{assump:psi_exists} holds, are available. Then \eqref{e:swsys} is GAS under every switching signal \(\sigma\in S_{\tau}\), where \(\tau\) is obtained from Algorithm \ref{algo:dwell_time}.
   \end{proposition}
   \begin{proof}
        Follows from Lemmas \ref{lem:auxres1} and \ref{lem:auxres2}.
   \end{proof}

   \begin{remark}
   \label{rem:tau_choice}
   \rm{
        Notice that the choice of \(\lambda_{s}\) for which the feasibility problem \eqref{e:feasprob} admits a solution for all \(i\in\P\), is not unique. Algorithm \ref{algo:dwell_time} exits Step II with the minimum such \(\lambda_{s}\in\{h,2h,\ldots,kh\}\). To minimize \(\tau\) over all \(\lambda_{s}\in\{h,2h,\ldots,kh\}\) such that there is a solution to \eqref{e:feasprob} for all \(i\in\P\), an algorithm requires to store all such \(\lambda_{s}\), compute the corresponding \(\mu\) and \(\tau\), and output the minimum value of \(\tau\).
   }
   \end{remark}

   \begin{remark}
   \label{rem:h_choice}
   \rm{
        An important aspect of Algorithm \ref{algo:dwell_time} is the choice of the step size \(h>0\). Fact \ref{fact:key1} and Lemma \ref{lem:auxres1} guarantee the existence of \(\lambda_{s}\in]0,1[\) such that \eqref{e:feasprob} admits solutions for all \(i\in\P\). However, \(\lambda_{s}\) may be very small. For executing Algorithm \ref{algo:dwell_time}, one may either pick \(h\) to be close to \(0\), or perform a trial and error procedure with \(h = 0.1,0.01,0.001,\ldots\) until a solution to \eqref{e:feasprob} is found for all \(i\in\P\).
   }
   \end{remark}

   \begin{remark}
   \label{rem:suff_condn}
   \rm{
        It is worth noting that Algorithm \ref{algo:dwell_time} computes \emph{a} stabilizing minimum dwell time and not \emph{the} minimum dwell time on every subsystem required for GAS of \eqref{e:swsys}. Indeed, Lemma \ref{lem:auxres2} provides a sufficient condition for GAS of \eqref{e:swsys} and does not conclude that a switching signal \(\sigma\) satisfying a dwell time \(\tau'<\tau\), where \(\tau\) obtained from Algorithm \ref{algo:dwell_time}, is destabilizing.
   }
   \end{remark}
\section{A numerical example}
\label{s:numex}
        Consider \(\P = \{1,2,3,4,5\}\). The matrices \(A_{i}\), \(i\in\P\) and their eigenvalues are given in Table \ref{tab:data_set1}. Our objective is to design a stabilizing minimum dwell time \(\tau\) when \(A_{i}\), \(i\in\P\) are not known, but subsystems data, \(\chi\) that satisfy Assumption \ref{assump:psi_exists} are available.\footnote{For this experiment, we generate the elements of the matrices \(A_{i}\), \(i\in\P\) with numbers from the interval \([-1,1]\) chosen uniformly at random. We build a simulation model, \(\M\), in Scilab 6.0.2 to generate the subsystems data \(\chi\).} We employ Algorithm \ref{algo:dwell_time} for this purpose. The following steps are executed:\\\\
        {Step I}: For each subsystem \(i\in\P\), we construct \(\Psi_{i}\) from the elements of the set \(\chi\). Numerical values of \(x_{i}(0),x_{i}(1),\cdots,x_{i}(L)\), \(i\in\P\) and their corresponding \(\Psi_{i}\), \(i\in\P\) are given in Table \ref{tab:data_set2}.\\\\
        {Step II}: We fix \(h = 0.1\), vary \(\lambda_{s}\) over the interval \(]0,1[\) with a step size \(h\), and solve the feasibility problem \eqref{e:feasprob} for symmetric and positive definite matrices, \(P_{i}\), \(i\in\P\). The following set of solutions is obtained:\footnote{The feasibility problem \eqref{e:feasprob} is solved using the lmisolver tool in Scilab 6.0.2.}
        \begin{align*}
            \lambda_{s} &= 0.7,\\
            \scriptstyle{P_{1}} &\scriptstyle{=}  \pmat{\scriptstyle{3750.4372} & \scriptstyle{-286.02836} &  \scriptstyle{74.534767} & \scriptstyle{-96.835359} &  \scriptstyle{286.10128}\\
                    \scriptstyle{-286.02836} &  \scriptstyle{1921.5791} & \scriptstyle{-303.73749} &  \scriptstyle{66.823217} &  \scriptstyle{99.566251}\\
                    \scriptstyle{74.534767} & \scriptstyle{-303.73749} &  \scriptstyle{1190.0666} & \scriptstyle{-197.08691} &  \scriptstyle{64.028826}\\
                    \scriptstyle{-96.835359} & \scriptstyle{66.823217} & \scriptstyle{-197.08691} &  \scriptstyle{693.16563} & \scriptstyle{-156.95092}\\
                    \scriptstyle{286.10128} &  \scriptstyle{99.566251} &  \scriptstyle{64.028826} & \scriptstyle{-156.95092} &  \scriptstyle{366.05467}},\\
            \scriptstyle{P_{2}} &\scriptstyle{=} \pmat{\scriptstyle{2618.2125} & \scriptstyle{-9.5389543} & \scriptstyle{-31.223599} & \scriptstyle{-207.59372} & \scriptstyle{-59.168975}\\
                            \scriptstyle{-9.5389543} &  \scriptstyle{1638.8187} & \scriptstyle{-24.685406} & \scriptstyle{-45.338756} & \scriptstyle{-93.56585}\\
                            \scriptstyle{-31.223599} & \scriptstyle{-24.685406} &  \scriptstyle{990.38532} & \scriptstyle{-12.006848} & \scriptstyle{-29.020649}\\
                            \scriptstyle{-207.59372} & \scriptstyle{-45.338756} & \scriptstyle{-12.006848} &  \scriptstyle{631.69632} & \scriptstyle{-11.33901}\\
                            \scriptstyle{-59.168975} & \scriptstyle{-93.56585} &  \scriptstyle{-29.020649} & \scriptstyle{-11.33901}  &  \scriptstyle{376.04863}},\\
            \scriptstyle{P_{3}} &\scriptstyle{=} \pmat{\scriptstyle{3740.8854} &  \scriptstyle{1180.8692} &  \scriptstyle{67.953807} &  \scriptstyle{402.11992} & \scriptstyle{-545.80508}\\
                            \scriptstyle{1180.8692} &  \scriptstyle{2263.7297} &  \scriptstyle{597.53107} & \scriptstyle{-94.37863} &   \scriptstyle{3.2384074}\\
                            \scriptstyle{67.953807} &  \scriptstyle{597.53107} &  \scriptstyle{1335.1877} &  \scriptstyle{278.34338} & \scriptstyle{-170.12372}\\
                            \scriptstyle{402.11992} & \scriptstyle{-94.37863}  &  \scriptstyle{278.34338} &  \scriptstyle{790.42701} &  \scriptstyle{36.921941}\\
                            \scriptstyle{-545.80508} &  \scriptstyle{3.2384074} & \scriptstyle{-170.12372} &  \scriptstyle{36.921941} &  \scriptstyle{524.71808}},\\
            \scriptstyle{P_{4}} &\scriptstyle{=} \pmat{\scriptstyle{3481.4063} & \scriptstyle{-1322.4505} & \scriptstyle{-349.63603} &  \scriptstyle{805.92625} &  \scriptstyle{5.8206481}\\
                        \scriptstyle{-1322.4505} &  \scriptstyle{2144.6438} & \scriptstyle{-369.98414} & \scriptstyle{-465.86262} &  \scriptstyle{186.60695}\\
                        \scriptstyle{-349.63603} & \scriptstyle{-369.98414} &  \scriptstyle{999.52474} & \scriptstyle{-208.06702} & \scriptstyle{-135.54574}\\
                        \scriptstyle{805.92625} &  \scriptstyle{-465.86262} & \scriptstyle{-208.06702} &  \scriptstyle{594.47594} & \scriptstyle{-20.831212}\\
                        \scriptstyle{5.8206481} &  \scriptstyle{186.60695} & \scriptstyle{-135.54574} & \scriptstyle{-20.831212} &  \scriptstyle{140.49607}},\\
            \scriptstyle{P_{5}} &\scriptstyle{=} \pmat{\scriptstyle{3521.175} &  \scriptstyle{-60.946413} & \scriptstyle{-304.62761} &  \scriptstyle{368.42809} & \scriptstyle{-457.01911}\\
                            \scriptstyle{-60.946413} &  \scriptstyle{1337.4811} & \scriptstyle{-130.53998} & \scriptstyle{-179.42523} &  \scriptstyle{0.623482}\\
                            \scriptstyle{-304.62761} & \scriptstyle{-130.53998} &  \scriptstyle{779.40748} & \scriptstyle{-138.61513} & \scriptstyle{-35.892171}\\
                            \scriptstyle{368.42809} & \scriptstyle{-179.42523} & \scriptstyle{-138.61513} &  \scriptstyle{489.07866} & \scriptstyle{-167.19363}\\
                            \scriptstyle{-457.01911} &  \scriptstyle{0.623482} &  \scriptstyle{-35.892171} & \scriptstyle{-167.19363} &  \scriptstyle{189.56074}}.
        \end{align*}
        {Step III}: We compute the scalars \(\mu_{ij}\), \(i,j\in\P\) by employing \eqref{e:mu_estimate} and then fix \(\displaystyle{\mu=\max_{i,j\in\P}\mu_{ij}}\). The numerical values are given below:
        \begin{align*}
            \mu_{11} &= 1,\:\:&\:\:\mu_{12} &= 1.8655187,\\
             \mu_{13} &= 3.2227957,\:\:&\:\:\mu_{14} &= 1.9351747,\\
              \mu_{15} &= 1.6117808,\:\:&\:\:\mu_{21} &= 1.7165712,\\
               \mu_{22} &= 1,\:\:&\:\:\mu_{23} &= 2.548444,\\
              \mu_{24} &= 2.6037244,\:\:&\:\:\mu_{25} &= 1.922591,\\
               \mu_{31} &= 6.2478964,\:\:&\:\:\mu_{32} &= 3.8349598,\\
                \mu_{33} &= 1,\:\:&\:\:\mu_{34} &= 3.7633396,\\
               \mu_{35} &= 2.3671962,\:\:&\:\:\mu_{41} &= 4.013124,\\
                \mu_{42} &= 4.024821,\:\:&\:\:\mu_{43} &= 6.6157071,\\
                 \mu_{44} &= 1,\:\:&\:\:\mu_{45} &= 3.1883122,\\
                \mu_{51} &= 6.7105711,\:\:&\:\:\mu_{52} &= 5.9626058,\\
                 \mu_{53} &= 9.4062392,\:\:&\:\:\mu_{54} &= 3.60176,\\
                  \mu_{55} &= 1,
        \end{align*}
        and
        \begin{align*}
            \mu &= 9.4062392.
        \end{align*}
        {Step IV}: We fix \(\varepsilon = 0.01\) and obtain a stabilizing minimum dwell time \(\tau = \lceil\frac{\ln\mu}{\abs{\ln\lambda_{s}}}+\varepsilon\rceil = 7\) units of time.
        \begin{table*}[htbp]
	\centering
    \scalebox{0.75}{
	\begin{tabular}{|c | c | c|}
		\hline
		\(i\) & \(A_{i}\) & \(\text{eigenvalues of}\:A_{i}\)\\
		\hline
		\(1\) & \(\pmat{0.2799379 & 0.0435507 & 0.0915753 & -0.1593086 &  0.2272202\\
                    1 & 0 & 0 & 0 & 0\\
                    0 & 1 & 0 & 0 & 0\\
                    0 & 0 & 1 & 0 & 0\\
                    0 & 0 & 0 & 1 & 0}\) & \(-0.5735302 \pm 0.5186809j\), \(0.7753572\), \(0.3258206 \pm 0.6196145j\)\\

		\hline
        \(2\) & \(\pmat{0.0204712 &  0.0840217  & 0.1088276  & 0.0248621 & -0.292626\\
                    1 & 0 & 0 & 0 & 0\\
                    0 & 1 & 0 & 0 & 0\\
                    0 & 0 & 1 & 0 & 0\\
                    0 & 0 & 0 & 1 & 0}\) & \(-0.7742757\), \(-0.2640693 \pm 0.7498403j\), \(0.6614427 \pm 0.4006231j\)\\
        \hline
        \(3\) & \(\pmat{-0.7060622 & -0.0678662 & -0.2441103  & 0.1226663 &  0.2980952\\
                    1 & 0 & 0 & 0 & 0\\
                    0 & 1 & 0 & 0 & 0\\
                    0 & 0 & 1 & 0 & 0\\
                    0 & 0 & 0 & 1 & 0}\) & \(-0.8014558 \pm 0.2326244j\), \(0.1198989 \pm 0.7981576j\), \(0.6570514\)\\
        \hline
        \(4\) & \(\pmat{0.6482512 & 0.0272578 & -0.4435161 & 0.1849962 & 0.053028\\
                    1 & 0 & 0 & 0 & 0\\
                    0 & 1 & 0 & 0 & 0\\
                    0 & 0 & 1 & 0 & 0\\
                    0 & 0 & 0 & 1 & 0}\) & \(0.4461056 \pm 0.6455056j\), \(0.6366064\), \(-0.6822682\), \(-0.1982982\)\\
        \hline
        \(5\) & \(\pmat{0.2486157 &  0.0809103 & -0.0931076 &  0.3252463 & -0.1238403\\
                    1 & 0 & 0 & 0 & 0\\
                    0 & 1 & 0 & 0 & 0\\
                    0 & 0 & 1 & 0 & 0\\
                    0 & 0 & 0 & 1 & 0}\) & \(-0.8266152\), \(-0.0028933 \pm 0.7301589j\), \(0.6460684\), \(0.4349492\)\\
        \hline
	\end{tabular}}
	\caption{Description of the subsystems.}\label{tab:data_set1}
	\end{table*}
    \begin{table*}[htbp]
	\centering
    \scalebox{0.8}{
	\begin{tabular}{|c | c | c|}
		\hline
		\(i\) & \(x_{i}(0),x_{i}(1),\cdots,x_{i}(L)\) & \(\Psi_{i}\)\\
		\hline
		\(1\) & \(\pmat{-0.3776165\\0.5511093\\-0.9545606\\0.4685422\\0.0824293},
            \pmat{-0.2250353\\-0.3776165\\0.5511093\\-0.9545606\\0.4685422},
            \pmat{0.2295586\\-0.2250353\\-0.3776165\\0.5511093\\-0.9545606},
            \pmat{-0.2848105\\0.2295586\\-0.2250353\\-0.3776165\\0.5511093},
            \pmat{0.0950412\\-0.2848105\\0.2295586\\-0.2250353\\-0.3776165},
            \pmat{-0.0147282\\0.0950412\\-0.2848105\\0.2295586\\-0.2250353}\) &
            \(\pmat{-0.2250353 & 0.2295586 & -0.2848105 & 0.0950412 & -0.0147282\\
                            -0.3776165 & -0.2250353 & 0.2295586 & -0.2848105 & 0.0950412\\
                            0.5511093 & -0.3776165 & -0.2250353 & 0.2295586 & -0.2848105\\
                            -0.9545606 & 0.5511093 & -0.3776165 & -0.2250353 & 0.2295586\\
                            0.4685422 & -0.9545606 & 0.5511093 & -0.3776165 & -0.2250353\\
                            0.0824293 & 0.4685422 & -0.9545606 & 0.5511093 & -0.3776165}\)\\
        \hline
        \(2\) & \(\pmat{0.5879499\\0.4187297\\0.8417496\\0.5093734\\0.6150621},
                    \pmat{-0.028495 \\   0.5879499 \\  0.4187297 \\  0.8417496 \\  0.5093734},
                    \pmat{-0.0337416 \\ -0.028495 \\   0.5879499 \\  0.4187297 \\  0.8417496},
                    \pmat{-0.1750071 \\ -0.0337416 \\ -0.028495  \\  0.5879499 \\  0.4187297},
                    \pmat{-0.1174322 \\ -0.1750071 \\ -0.0337416 \\ -0.028495  \\  0.5879499},
                    \pmat{-0.1935383 \\ -0.1174322 \\ -0.1750071 \\ -0.0337416 \\ -0.028495}\) &
                    \(\pmat{-0.028495  & -0.0337416 & -0.1750071 & -0.1174322 & -0.1935383\\
                        0.5879499 & -0.028495 &  -0.0337416 & -0.1750071 & -0.1174322\\
                        0.4187297 &  0.5879499 & -0.028495 &  -0.0337416 & -0.1750071\\
                        0.8417496 &  0.4187297 &  0.5879499 & -0.028495 &  -0.0337416\\
                        0.5093734 &  0.8417496 &  0.4187297 &  0.5879499 & -0.028495\\
                        0.6150621 &  0.5093734 &  0.8417496 &  0.4187297 &  0.5879499}\)\\
        \hline
        \(3\) & \(\pmat{-0.6008976 \\  0.0264991 \\ -0.7239973 \\ -0.1963839 \\ -0.2353341},
                  \pmat{0.5049662 \\ -0.6008976 \\  0.0264991 \\ -0.7239973 \\ -0.1963839},
                  \pmat{-0.4695768 \\  0.5049662 \\ -0.6008976 \\  0.0264991 \\ -0.7239973},
                  \pmat{0.231396 \\  -0.4695768 \\  0.5049662 \\ -0.6008976 \\  0.0264991},
                  \pmat{-0.3205896 \\  0.231396 \\  -0.4695768 \\  0.5049662 \\ -0.6008976},
                  \pmat{0.2080984 \\ -0.3205896 \\  0.231396 \\  -0.4695768 \\  0.5049662}\) &
                \(\pmat{0.5049662 & -0.4695768 &  0.231396 &  -0.3205896 &  0.2080984\\
                        -0.6008976 &  0.5049662 & -0.4695768 &  0.231396 &  -0.3205896\\
                        0.0264991 & -0.6008976 &  0.5049662 & -0.4695768 &  0.231396\\
                        -0.7239973  & 0.0264991 &  -0.6008976  & 0.5049662 & -0.4695768\\
                        -0.1963839 & -0.7239973 &  0.0264991 & -0.6008976 &  0.5049662\\
                        -0.2353341 & -0.1963839 & -0.7239973  & 0.0264991 & -0.6008976}\)\\
        \hline
        \(4\) & \(\pmat{-0.5687414 \\ -0.6945576 \\  0.0805042 \\ -0.3177508 \\  0.4460485},
                \pmat{-0.4584539 \\ -0.5687414 \\ -0.6945576 \\  0.0805042 \\ -0.3177508},
                \pmat{-0.0066052 \\ -0.4584539 \\ -0.5687414 \\ -0.6945576 \\  0.0805042},
                \pmat{0.1112462 \\ -0.0066052 \\ -0.4584539 \\ -0.5687414 \\ -0.6945576},
                \pmat{0.1332211 \\  0.1112462 \\ -0.0066052 \\ -0.4584539 \\ -0.5687414},
                \pmat{-0.0226489 \\  0.1332211 \\  0.1112462 \\ -0.0066052 \\ -0.4584539}\) &
                \(\pmat{-0.4584539 & -0.0066052 &  0.1112462 &  0.1332211 & -0.0226489\\
                -0.5687414 & -0.4584539 & -0.0066052 &  0.1112462 &  0.1332211\\
                -0.6945576 & -0.5687414 & -0.4584539 & -0.0066052 &  0.1112462\\
                0.0805042 & -0.6945576 & -0.5687414 & -0.4584539 & -0.0066052\\
                -0.3177508 &  0.0805042 & -0.6945576 & -0.5687414 & -0.4584539\\
                0.4460485 & -0.3177508 &  0.0805042 & -0.6945576 & -0.5687414}\)\\
        \hline
        \(5\) & \(\pmat{-0.1540222 \\ -0.2083555 \\ -0.56438  \\  -0.2037382 \\ -0.8053248},
                  \pmat{0.0308642 \\ -0.1540222 \\ -0.2083555 \\ -0.56438 \\   -0.2037382},
                  \pmat{-0.1437207 \\  0.0308642 \\ -0.1540222 \\ -0.2083555 \\ -0.56438},
                  \pmat{-0.0167672 \\ -0.1437207 \\  0.0308642 \\ -0.1540222 \\ -0.2083555},
                  \pmat{-0.0429631 \\ -0.0167672 \\ -0.1437207 \\  0.0308642 \\ -0.1540222},
                  \pmat{0.0304562 \\ -0.0429631 \\ -0.0167672 \\ -0.1437207 \\  0.0308642}\) &
                \(\pmat{0.0308642 & -0.1437207 & -0.0167672 & -0.0429631 &  0.0304562\\
                    -0.1540222 &  0.0308642 & -0.1437207 & -0.0167672 & -0.0429631\\
                    -0.2083555 & -0.1540222 &  0.0308642 & -0.1437207 & -0.0167672\\
                    -0.56438  &  -0.2083555 & -0.1540222 &  0.0308642 & -0.1437207\\
                    -0.2037382 & -0.56438  &  -0.2083555 & -0.1540222 &  0.0308642\\
                    -0.8053248 & -0.2037382 & -0.56438 &  -0.2083555 & -0.1540222}\)\\
        \hline
	\end{tabular}}
	\caption{Description of subsystems data and their corresponding \(\Psi_{i}\), \(i\in\P\).}\label{tab:data_set2}
	\end{table*}

    We now demonstrate that \(\tau\) is indeed a stabilizing minimum dwell time for the switched system under consideration. Towards this end, we perform the following experiment: we pick \(x(0)\in\R^{5}\) from the interval \([-1,1]^{5}\) uniformly at random, design a switching signal \(\sigma\) randomly but ensuring a minimum dwell time \(\tau\) on every subsystem, and plot the state trajectory \((\norm{x(t)})_{t\in\N_{0}}\) of the switched system \eqref{e:swsys}. This process is repeated 1000 times, and the corresponding plots of \((\norm{x(t)})_{t\in\N_{0}}\) are given in Figure \ref{fig:x_plot}.
   \begin{figure*}
   \centering
        \includegraphics[height = 6cm, width = 10cm]{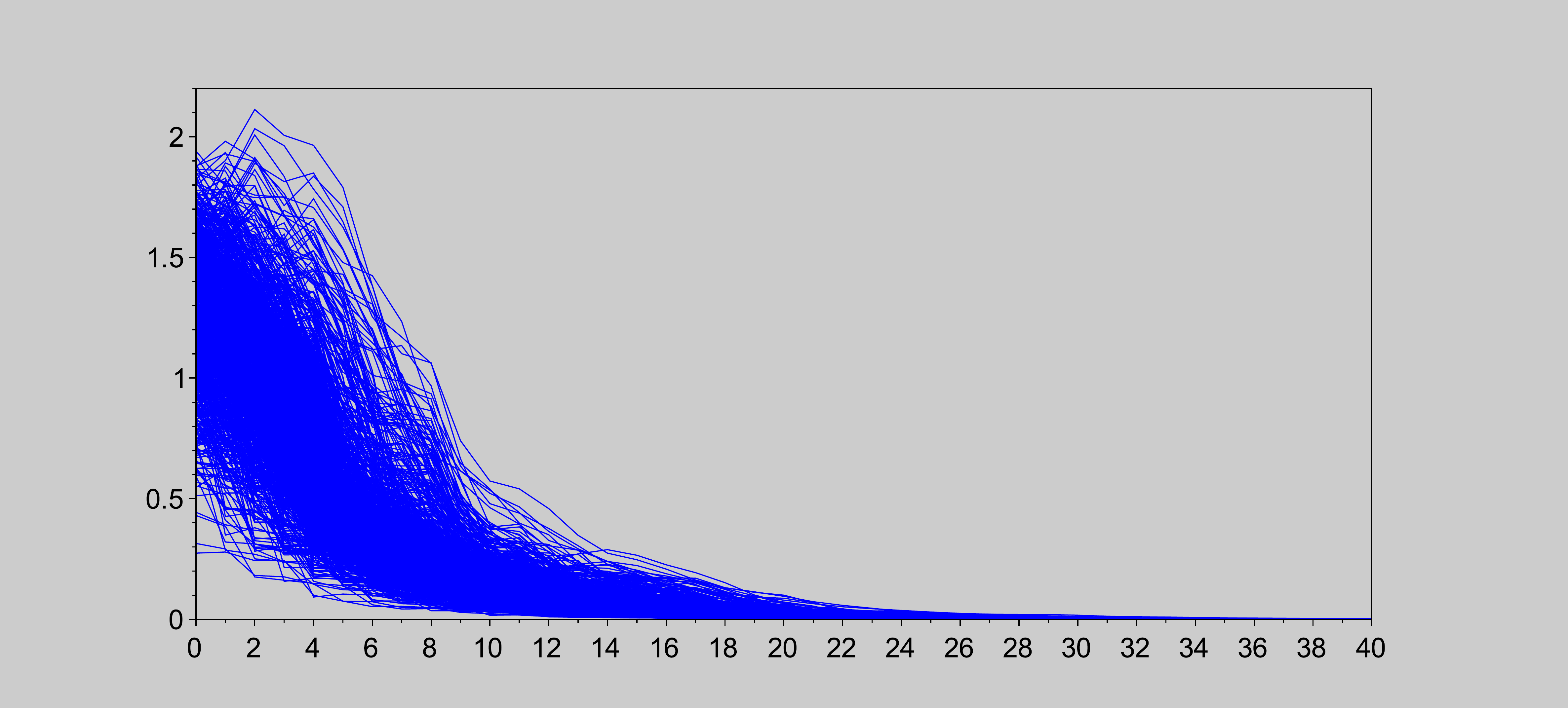}
        \caption{Plot of \((\norm{x(t)})_{t\in\N_{0}}\).}\label{fig:x_plot}
   \end{figure*}
\section{Conclusion}
\label{s:concln}
    To summarize, in this paper we presented an algorithm to compute stabilizing minimum dwell times for discrete-time switched linear systems when state-space models of their subsystems are not known explicitly. We consider that a set of finite traces of state trajectories of the subsystems that satisfies certain properties, is available. We design multiple Lyapunov functions corresponding to the subsystems from the subsystems data and determine a stabilizing minimum dwell time as a function of a set of scalars obtained from these Lyapunov functions.

    A next natural question is regarding the extension of our techniques to the design of stabilizing switching signals when not all subsystems are stable and the admissible switches between the subsystems are restricted. This matter is currently under investigation, and will be reported elsewhere.


\end{document}